%% file: main.tex
\documentclass[conference]{IEEEtran}
\usepackage{soul}

\newcommand{\edit}[1]{{\color{black} #1}}

\IEEEoverridecommandlockouts
\input{my_preamble.tex}
\usepackage{amsmath,amssymb,amsfonts}
\usepackage{array}
\usepackage[caption=false,font=normalsize,labelfont=sf,textfont=sf]{subfig}
\usepackage{textcomp}
\usepackage{stfloats}
\usepackage{url}
\usepackage{verbatim}
\usepackage{graphicx}
\usepackage{cite}
\usepackage{color}
\usepackage[top=0.75in, bottom=1.05in, left=0.625in, right=0.625in, columnsep=0.21in]{geometry}
\usepackage[acronym,shortcuts]{glossaries}
\newacronym{3GPP}{3GPP}{3rd Generation Partnership Project}
\newacronym{AWGN}{AWGN}{additive white Gaussian noise}
\newacronym{BS}{BS}{base station}
\newacronym{CDF}{CDF}{cumulative distribution function}
\newacronym{CSCG}{CSCG}{circularly symmetric complex Gaussian}
\newacronym{CSI}{CSI}{channel state information}
\newacronym{iid}{i.i.d.}{independent and identically distributed}
\newacronym{LoS}{LoS}{line-of-sight}
\newacronym{MIMO}{MIMO}{multiple-input multiple-output}
\newacronym{MMSE}{MMSE}{minimum mean-squared error}
\newacronym{NCR}{NCR}{network-controlled repeater}
\newacronym{NLoS}{NLoS}{non line-of-sight}
\newacronym{PDF}{PDF}{probability density function}
\newacronym{RAMIMO}{RA-MIMO}{repeater-assisted massive MIMO}
\newacronym{SE}{SE}{spectral efficiency}
\newacronym{SINR}{SINR}{signal-to-interference-plus-noise ratio}
\newacronym{TDD}{TDD}{time-division duplex}
\newacronym{UE}{UE}{user equipment}

\begin{document}

\title{Repeater-Assisted Massive MIMO Downlink Performance with Calibration Errors\vspace*{-.3\baselineskip}}
\author{\IEEEauthorblockN{Kohei Ueda\IEEEauthorrefmark{1}, Anubhab Chowdhury\IEEEauthorrefmark{2}, Koji Ishibashi\IEEEauthorrefmark{1}, and Erik G. Larsson\IEEEauthorrefmark{2}}\\
\IEEEauthorblockA{\IEEEauthorrefmark{1}Advanced Wireless \& Communication Research Center (AWCC), The University of Electro-Communications, \\1-5-1 Chofugaoka, Chofu, Tokyo 182-8585, Japan\\
\IEEEauthorrefmark{2}Department of Electrical Engineering (ISY), Linköping University, 58183 Linköping, Sweden\\
  Emails: \IEEEauthorrefmark{1}k.ueda@awcc.uec.ac.jp,  \IEEEauthorrefmark{2}anuch87@liu.se, \IEEEauthorrefmark{1}koji@ieee.org, and \IEEEauthorrefmark{2}erik.g.larsson@liu.se
\vspace*{-1.3\baselineskip}}
\thanks{{This work was supported in part by the KAW foundation, ELLIIT, the Swedish Research Council (VR), JST CRONOS, Japan Grant Number JPMJCS24N1, and JST BOOST, Japan Grant Number JPMJBS2415.}}
}



\maketitle

\begin{abstract}
Reciprocity-based downlink beamforming is imperative for a scalable time-division duplex massive multiple-input multiple-output~(MIMO) deployment. Specifically, for a dual-antenna repeater-assisted massive MIMO system, a mismatch between forward and reverse path gains at the repeater can exacerbate the overall calibration error between the user equipments (UEs) and the base station (BS), which potentially also contains calibration errors of their individual radio-frequency chains. This paper models the effects of such calibration errors, underpins the relations between the uplink and downlink channels for repeater-assisted systems with calibration errors clubbed with the over-the-air channel estimation errors, and derives analytical expressions of the downlink spectral efficiency. The presented results can then be simplified to several special cases, underscoring situations wherein such errors can become pronounced.
\end{abstract}


\section{Introduction}


{Repeaters are full-duplex devices that instantaneously receive and retransmit signals within a given band, in a way transparent to the network, acting as ordinary channel scatterers but with amplification.}
Repeaters can be deployed on a large scale to mitigate coverage holes for cellular \gls{MIMO} wireless systems~\cite{Sara_Vieira_Erik_Mag, Byoung_Magazine, Sergiy_WCL, Ozan_repeaters, 2025_Jiannan_repeater, 2026_ueda_RAMIMO_TDDvsFDD}.
In \cite{Sara_Vieira_Erik_Mag}, it is demonstrated that repeaters procure \gls{SE}  similar to distributed \gls{MIMO} while obviating the front-haul signaling overhead required by the latter. Specifically, dual-antenna repeater technology has been included in the 5G NR standards since \gls{3GPP} Release $18$~\cite{3gpp}, named as \gls{NCR}.

For dual-antenna repeaters to appear transparent to the \gls{BS} and the \glspl{UE}, they need to be reciprocity calibrated, i.e., the forward and the reverse path gains should be equal.
In this regard,~\cite{Erik_WCL} developed a non-linear least-squares framework to estimate the ratio of the forward to reverse path gains of a repeater, following which a compensation for the mismatch in the uplink/downlink antenna path gains can be executed.
Further, the propounded strategy in~\cite{Erik_WCL} also extends to multiple repeaters. 

Nonetheless, the over-the-air calibration techniques heavily rely on the channel estimates, which in turn are prone to errors. When reciprocity-based beamforming is applied in the downlink massive MIMO system, the calibration error degrades the beamforming gain, and consequently the overall downlink \gls{SE}~\cite{Luo_Calibration, Adve_TCOM, 2018_Emil_reciprocityError, Ribhu_CRM_EGL_TSP}.
With regard to repeater-assisted systems,~\cite{Yiming_PIMRC} investigated the effects of mismatched forward and reverse link gains at the repeaters on the downlink \gls{SINR}. However, the results therein ignore the direct link channel between the \gls{BS} and the \acp{UE}, and are restricted to errors only in the repeaters' calibration coefficients, and highly depend on large-dimensional approximations.  

Our proposed framework considers both channel estimation and calibration errors in all transceiver chains
between the \gls{BS} and the \acp{UE}, i.e., the direct link channels and the channels via the repeater. Thus, the model presented in~\cite{Yiming_PIMRC} can be subsumed as a special case of our analysis. 
{We obtain expressions for the downlink \ac{SE} via achievable rate analysis for finite numbers of antennas. These results quantify} the composite effects of calibration errors at the \gls{BS}, the \acp{UE}, and the repeaters. 
\par
\emph{Notation:}~The identity matrix of size $N$ and the all-zero column vector of size $N$ are denoted by $\mathbf{I}_N$ and $\mathbf{0}_N$, respectively.
A \ac{CSCG} with mean vector $\boldsymbol{\mu}$ and covariance matrix $\boldsymbol{\Sigma}$ is denoted by $\mathcal{CN}(\boldsymbol{\mu},\boldsymbol{\Sigma})$.
The operators $(\cdot)^{*}$, $(\cdot)^{\mathrm{T}}$, $(\cdot)^{\mathrm{H}}$, and $\odot$ denote conjugate, transpose, Hermitian transpose, and element-wise product, respectively.
The operators $\mathbb{E}[\cdot]$, $\mathrm{Var}(\cdot)$, $\mathrm{Tr}\{\cdot\}$ and $\|\cdot\|$ are expectation, variance, trace of a matrix, and $\ell_2$-norm, respectively.
For a diagonal matrix $\mathbf{D}$, $|\mathbf{D}|$ implies element-wise absolute values.
%
%
%
\section{System Model}
\begin{figure}[t]
	\centering
		\begin{tikzpicture}
		\newdimen\R
		\R=2cm
		\draw (-4,3) pic (BSt) {bs} node [xshift = 0mm, yshift=12mm] {$\mathbf{t}_{\mathrm{dl}}\in\mathbb{C}^{N}$};
		\draw (1.7,.7) pic (UE) {user};
		\draw (2.4,1.7) pic (UE2) {user};		
		\draw (2,3.1) pic (UE3) {user} node [yshift=-8mm, xshift=2.3mm]{~$k$-th User~$\bbr{t_{\mathrm{U}, k}, r_{\mathrm{U}, k}}$};
		\draw (-2,1) pic (R) {repeater} node [yshift=-4mm]{Repeater: \{$\alpha_{\sfu},\alpha_{\sfd}\}$};
		\node  [align=center,below=-4mm of BStdim.south] {\gls{MIMO} \gls{BS}\\ $\bbr{\mathbf{T}_{\mathrm{B}},\mathbf{R}_{\rm B}}$};
		
		\draw[-{Stealth[length=3.5mm,  open, round]},black, dashed, thick,shorten >= 0.3cm] (-4,3.8) to (-2.2,1.7) node [xshift=-5mm, yshift=13mm]{$\mathbf{g}^{\rm T}$};
		
		\draw[-{Stealth[length=3.5mm,  open, round]},black, dashed, thick,shorten >= 0.3cm] (-4,3.8) to (2,3.1) node [xshift=-15mm, yshift=0mm]{$\mathbf{h}_{k}$};
		
		\draw[-{Stealth[length=3.5mm,  open, round]},black, dashed, thick,shorten >= 0.3cm] (-1.4,1.9) to (2,3.1) node [xshift=-20mm, yshift=-4mm]{$f_{k}$};

\draw[-{Stealth[length=3.5mm,  open, round]},black, dashed, thick,shorten >= 0.3cm] (-1.4,1.9) to (2.4,1.7);

\draw[-{Stealth[length=3.5mm,  open, round]},black, dashed, thick,shorten >= 0.3cm] (-1.4,1.9) to (1.7,1);
		\end{tikzpicture}
	\caption{Repeater-assisted communication. Dashed lines~($--$) indicate over-the-air channels that are reciprocal. Solid lines~($-$) in blue and red along the repeater show the forward and reverse path gains of the transceiver chains.}\label{fig: system_model
	}
\end{figure}
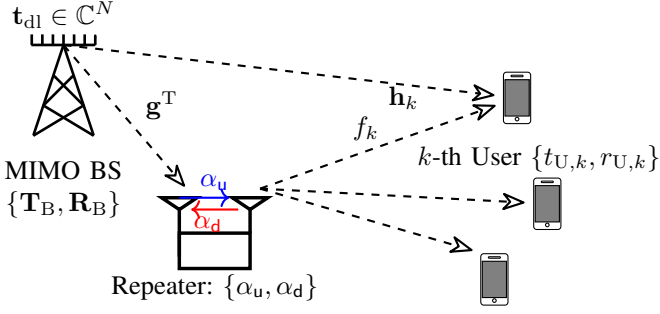

We consider a dual-antenna \gls{RAMIMO} system where $K$ \acp{UE} are being served by a massive \gls{MIMO} \gls{BS}~\cite{Erik_WCL}, operating in \gls{TDD}. The \gls{BS} is equipped with $N$ antennas, and \acp{UE} are all single-antenna. 
The diagonal matrices $\mathbf{T}_{\rm B}$ and $\mathbf{R}_{\rm B}\in\mathbb{C}^{N\times N}$ consist of the transmit and receive reciprocity coefficients, with the $n$-th diagonal element being $\left\{t_{\mathrm{B}, n}\right\}$ and $\left\{r_{\mathrm{B}, n}\right\}$, corresponding to the $n$-th \gls{BS} antenna, respectively.
Similarly, the corresponding coefficients for the $k$-th \gls{UE} are denoted by $\left\{t_{\mathrm{U}, k}\right\}$ and $\left\{r_{\mathrm{U}, k}\right\}$~\cite{Vieira_TWC}.
Let $\alpha_{\mathrm{u}}$ and $\alpha_{\mathrm{d}}$ be the forward and reverse link complex-valued gains of the repeater~\cite{Erik_WCL}.
Next, the over-the-air (uplink) channels from the $k$-th \gls{UE} to the \gls{BS}, the repeater to the \gls{BS}, and the repeater to the $k$-th \gls{UE} are denoted by $\mathbf{h}_{k}\in\mathbb{C}^{N\times 1}$, $\mathbf{g}\in\mathbb{C}^{N\times 1}$, and $f_{k}\in\mathbb{C}$, respectively. These over-the-air channels are reciprocal. Thus, the equivalent uplink channel from the $k$-th \gls{UE} to the \gls{BS} can be written as:
\begin{align}
\mathbf{h}_{\mathrm{ul}, k}=\bar{\mathbf{h}}_{\mathrm{ul}, k}+\bar{\mathbf{g}}_{\mathrm{ul},k},
\end{align}
where the direct link equivalent channel $\bar{\mathbf{h}}_{\mathrm{ul},k}$ and the over-the-air uplink channel $\mathbf{h}_{k}$ are  related as 
\begin{align}
    \bar{\mathbf{h}}_{\mathrm{ul}, k}=t_{\mathrm{U}, k}\mathbf{R}_{\rm B}\mathbf{h}_{k}.
\end{align}
We note that the over-the-air channel $\mathbf{h}_{k}$ is measured with the repeater being turned off; however, it can contain (reciprocal) passive scattering effects of the repeater. The equivalent uplink channel from the \gls{UE} to the \gls{BS} via the repeater $\bar{\mathbf{g}}_{\mathrm{ul},k}$ is
\begin{align}
    \bar{\mathbf{g}}_{\mathrm{ul},k}=\alpha_{\mathrm{u}} t_{\mathrm{U}, k}\mathbf{R}_{\rm B}\left\{\mathbf{g} f_{k}\right\}.
\end{align}
Next, the equivalent downlink channel for the $k$-th \gls{UE} can be expressed as
\begin{align}
\mathbf{h}_{\mathrm{dl}, k}^{\rm T} &= r_{\mathrm{U}, k}\mathbf{}\mathbf{h}_{k}^{\rm T}\mathbf{T}_{\mathrm{B}}+\alpha_{\mathrm{d}} r_{\mathrm{U}, k}\left\{f_{k}\mathbf{g}^{\rm T}\right\}\mathbf{T}_{\mathrm{B}}\notag\\
&=\frac{r_{\mathrm{U}, k}}{t_{\mathrm{U}, k}}\left(\bar{\mathbf{h}}_{\mathrm{ul}, k}^{\rm T}+\frac{\alpha_{\mathrm{d}}}{\alpha_{\mathrm{u}}}\bar{\mathbf{g}}_{\mathrm{ul},k}^{\rm T}\right)\mathbf{R}_{\rm B}^{-1}\mathbf{T}_{\mathrm{B}},\label{eq: ul_dl}
\end{align}
where we used the fact that $\mathbf{T}_{\rm B}$ and $\mathbf{R}_{\rm B}$ are diagonal. To capture the mismatches in the calibration coefficients, we define $e_{\mathrm{U}, k}=\frac{r_{\mathrm{U}, k}}{t_{\mathrm{U}, k}}$, $\mathbf{E}_{\mathrm{B}}=\mathbf{R}_{\rm B}^{-1}\mathbf{T }_{\mathrm{B}}$ with the $n$-th diagonal element being $e_{\mathrm{B}, n}$, and $e_{\mathrm{R}}=\frac{\alpha_{\mathrm{d}}}{\alpha_{\mathrm{u}}}$. Thus, after taking transpose of~\eqref{eq: ul_dl}, we get
\begin{align}
\mathbf{h}_{\mathrm{dl}, k}=e_{\mathrm{U}, k}\mathbf{E}_{\mathrm{B}}\left(\bar{\mathbf{h}}_{\mathrm{ul}, k}+e_{\mathrm{R}}\bar{\mathbf{g}}_{\mathrm{ul},k}\right),
\end{align}
which establishes the relation between the equivalent uplink and downlink channels for repeater-assisted systems with reciprocity imperfections.
\subsubsection{Modeling Calibration Errors} In practice, $e_{\mathrm{U}, k}$, $\mathbf{E}_{\mathrm{B}}$, and $e_{\mathrm{R}}$ need to be calibrated and compensated~\cite{Erik_WCL}, so that the repeater acts transparent to the \acp{UE} and the \gls{BS}, and \gls{BS} can use reciprocity-based beamforming. 
However, the calibration errors cannot be avoided due to the noisy measurements during the over-the-air calibration process.
We capture such errors in calibration as \cite{Vieira_TWC}: $e_{\mathrm{U}, k}=\hat{e}_{\mathrm{U}, k}+\tilde{e}_{\mathrm{U}, k}$, $\mathbf{E}_{\mathrm{B}}=\hat{\mathbf{E}}_{\mathrm{B}}+\tilde{\mathbf{E}}_{\mathrm{B}}$, and $e_{\mathrm{R}}=\hat{e}_{\mathrm{R}}+\tilde{e}_{\mathrm{R}}$, where $\hat{e}_{\mathrm{U}, k}$, $\hat{\mathbf{E}}_{\mathrm{B}}$, and $\hat{e}_{\mathrm{R}}$ are the estimated ``mismatched" coefficients and assumed to be \emph{known}~\cite{Adve_TCOM}. The errors in the calibration are captured by $\tilde{e}_{\mathrm{U}, k}$, $\tilde{\mathbf{E}}_{\mathrm{B}}$, and $\tilde{e}_{\mathrm{R}}$, which are zero mean random variables with variances/covariance being $\sigma^2_{\mathrm{U}, k}$, $\sigma^2_{\mathrm{B}}\vI_{N}$, and $\sigma^2_{\mathrm{R}}$, respectively. Also, calibration errors across \acp{UE} are independent.

\subsubsection{Impact on Uplink Channel Estimation}
As a necessary first step to understand the effects of calibration error on system performance, we consider a noiseless repeater and abstract the channel estimation procedure. We consider $\bar{\mathbf{h}}_{\mathrm{ul}, k}=\hat{\bar{\mathbf{h}}}_{\mathrm{ul}, k}+\tilde{\bar{\mathbf{h}}}_{\mathrm{ul}, k}$ and $\bar{\mathbf{g}}_{\mathrm{ul},k}=\hat{\bar{\mathbf{g}}}_{\mathrm{ul},k}+\tilde{\bar{\mathbf{g}}}_{\mathrm{ul},k}$, where $\hat{\mathbf{x}}$ denotes the estimate of a variable $\mathbf{x}$, and $\tilde{\mathbf{x}}$ the estimation error. Then, the equivalent downlink channel in~\eqref{eq: dl_channel} becomes
\begin{align}
\mathbf{h}_{\mathrm{dl}, k}=&\left(\hat{e}_{\mathrm{U}, k}+\tilde{e}_{\mathrm{U}, k}\right)\left(\hat{\mathbf{E}}_{\mathrm{B}}+\tilde{\mathbf{E}}_{\mathrm{B}}\right)\left(\left\{\hat{\bar{\mathbf{h}}}_{\mathrm{ul}, k}+\tilde{\bar{\mathbf{h}}}_{\mathrm{ul}, k}\right\}\right.\notag\\
&\left.+\left(\hat{e}_{\mathrm{R}}+\tilde{e}_{\mathrm{R}}\right)\left\{\hat{\bar{\mathbf{g}}}_{\mathrm{ul},k}+\tilde{\bar{\mathbf{g}}}_{\mathrm{ul},k}\right\}\right)\triangleq\hat{\mathbf{h}}_{\mathrm{dl}, k}+\tilde{\mathbf{h}}_{\mathrm{dl}, k},\label{eq: dl_channel}
\end{align}
where, with a little algebra, the estimated equivalent downlink channel and the error vector can respectively be expressed as:
\begin{subequations}
\begin{align}
& \hat{\mathbf{h}}_{\mathrm{dl}, k}=\hat{e}_{\mathrm{U}, k}\hat{\mathbf{E}}_{\mathrm{B}}\left(\hat{\bar{\mathbf{h}}}_{\mathrm{ul}, k}+\hat{e}_{\mathrm{R}}\hat{\bar{\mathbf{g}}}_{\mathrm{ul},k}\right) = \hat{e}_{\mathrm{U}, k}\hat{\mathbf{E}}_{\mathrm{B}}\hat{\mathbf{h}}_{\mathrm{ul},k} \label{eq: hat_dl_effective}\\
&\tilde{\mathbf{h}}_{\mathrm{dl}, k}=\tilde{e}_{\mathrm{U}, k}\hat{\mathbf{E}}_{\mathrm{B}}\hat{\mathbf{h}}_{\mathrm{ul}, k}+{e}_{\mathrm{U}, k}(\tilde{\mathbf{E}}_{\mathrm{B}}\tilde{\mathbf{h}}_{\mathrm{ul}, k}+\hat{\mathbf{E}}_{\mathrm{B}}\tilde{\mathbf{h}}_{\mathrm{ul}, k}+\tilde{\mathbf{E}}_{\mathrm{B}}\hat{\mathbf{h}}_{\mathrm{ul}, k}),\label{eq: tilde_dl_effective}
\end{align}
\end{subequations} 
where $\hat{\mathbf{h}}_{\mathrm{ul}, k}\triangleq\hat{\bar{\mathbf{h}}}_{\mathrm{ul}, k}+\hat{e}_{\mathrm{R}}\hat{\bar{\mathbf{g}}}_{\mathrm{ul},k}$ and $\tilde{\mathbf{h}}_{\mathrm{ul}, k}\triangleq \tilde{\bar{\mathbf{h}}}_{\mathrm{ul}, k}+{e}_{\mathrm{R}}\tilde{\bar{\mathbf{g}}}_{\mathrm{ul},k}+\tilde{e}_{\mathrm{R}}\hat{\bar{\mathbf{g}}}_{\mathrm{ul},k}.$
It is clear from~\eqref{eq: tilde_dl_effective} that $ \tilde{\mathbf{h}}_{\mathrm{dl}, k}$ is also affected by the uplink estimated channels as well as the errors due to the multiplicative nature of the calibration coefficients. Thus, a calibration error can potentially increase the variance of the error in the effective estimated downlink channel, thereby degrading the \gls{SE}, which is discussed next. 
\section{Effects of  Calibration Error on Downlink \ac{SE}}
{In this section, we analyze the ergodic capacity with respect to channels, estimation errors of the channel and calibration coefficient, and noise, conditioned on a given estimate of the calibration coefficients.}
Using $\hat{\mathbf{h}}_{\mathrm{dl}, k}$ in~\eqref{eq: hat_dl_effective}, the signal received at the $k$-th \gls{UE} can be expressed as:
\begin{align}
\hat{s}_{k}=\sqrt{\rho\kappa_{k}}\mathbf{h}_{\mathrm{dl}, k}^{\rm T}\hat{\mathbf{h}}_{\mathrm{dl}, k}^{*}s_{k}+\sum_{n\neq k}\sqrt{\rho\kappa_{n}}\mathbf{h}_{\mathrm{dl}, k}^{\rm T}\hat{\mathbf{h}}_{\mathrm{dl}, n}^{*}s_{n}+w_{k},
\end{align}
where $\rho$ is the total downlink transmit power, $\kappa_{k}=\frac{\eta_k}{\mathrm{Tr}\{\mathbb{E}[\hat{\mathbf{h}}_{\mathrm{dl},k}\hat{\mathbf{h}}_{\mathrm{dl},k}^{\rm H}]\}},~\eta_k\in\bsr{0,1}$ is the power allocation variable including channel gain normalization of the $k$-th \gls{UE}, and $s_{k}$ is the unit-energy zero-mean data symbol intended for the $k$-th \gls{UE}.
$w_{k}\sim\mathcal{CN}(0, \sigma_{\rm UE}^2)$ is the \gls{AWGN} at the $k$-th \gls{UE}.
Further, we can rewrite $\hat{s}_{k}$, following the principle of use-and-then-forget bound~\cite{Marzetta_Larsson_Yang_Ngo_2016}, as 
\begin{align}
\hat{s}_{k}=\sqrt{\rho\kappa_{k}}\,\mathbb{E}\left[\hat{\mathbf{h}}_{\mathrm{dl}, k}^{\rm T}\hat{\mathbf{h}}_{\mathrm{dl}, k}^{*}\right]s_{k}+z_{k},
\end{align}
\edit{where the expectation is taken with respect to the randomness in the channels and the calibration errors.}
Then, the effective interference-plus-\gls{AWGN} $z_{k}$ constitutes
\begin{align}
z_{k}&=\sqrt{\rho\kappa_{k}}\left(\hat{\mathbf{h}}_{\mathrm{dl}, k}^{\rm T}\hat{\mathbf{h}}_{\mathrm{dl}, k}^{*}-\mathbb{E}\left[\hat{\mathbf{h}}_{\mathrm{dl}, k}^{\rm T}\hat{\mathbf{h}}_{\mathrm{dl}, k}^{*}\right]\right)s_{k}\\
&\hspace{2ex}+\sqrt{\rho\kappa_{k}}{\tilde{\mathbf{h}}_{\mathrm{dl}, k}^{\rm T}\hat{\mathbf{h}}_{\mathrm{dl}, k}^{*}}s_{k}+\sum_{n \neq k}\sqrt{\rho\kappa_{n}}\mathbf{h}_{\mathrm{dl}, k}^{\rm T}\hat{\mathbf{h}}_{\mathrm{dl}, n}^{*}s_{n}+w_{k}.\nonumber
\end{align}
\begin{prop}\label{prop_uncorrelated}
If the estimated uplink channels $\hat{\bar{\mathbf{h}}}_{\mathrm{ul}, k}$ and $\hat{\bar{\mathbf{g}}}_{\mathrm{ul},k}$ are uncorrelated with the errors $\tilde{\bar{\mathbf{h}}}_{\mathrm{ul}, k}$ and $\tilde{\bar{\mathbf{g}}}_{\mathrm{ul},k}$, it can be shown that $\hat{\mathbf{h}}_{\mathrm{dl}, k}$ is uncorrelated with $\tilde{\mathbf{h}}_{\mathrm{dl}, k}$; and consequently, $z_{k}$ is uncorrelated with $\mathbb{E}[\hat{\mathbf{h}}_{\mathrm{dl}, k}^{\rm T}\hat{\mathbf{h}}_{\mathrm{dl}, k}^{*}]s_{k}$.
\end{prop}
{Next, thanks to~\cref{prop_uncorrelated}, a use-and-then-forget type~\cite{Marzetta_Larsson_Yang_Ngo_2016}} lower bound on the downlink \gls{SE} of the $k$-th \gls{UE} can be evaluated as $\mathrm{R}_{k}=\log\left(1+\mathrm{SINR}_{k}\right)$, where the downlink \gls{SINR} is computed as $\mathrm{SINR}_{k}\triangleq {\abs{\sqrt{\rho\kappa_{k}} \, \mathbb{E}[\hat{\mathbf{h}}_{\mathrm{dl}, k}^{\rm T}\hat{\mathbf{h}}_{\mathrm{dl}, k}^{*}]}^2}/{\mathbb{E}[\abs{z_{k}}^2]}$. This, we can further break down as follows:
\begin{align}
\mathrm{SINR}_{k}=\frac{\mathrm{G}_k}{\mathrm{BI}_{k}+\mathrm{CI}_{k}+\mathrm{MI}_{k}+\sigma_{\rm UE}^2},\label{eq: SINR_k}
\end{align}
where the terms are defined as
\begin{subequations}
\begin{align}
&\mathrm{G}_k{\triangleq\rho\kappa_{k}\left|\mathbb{E}\left[\hat{\mathbf{h}}_{\mathrm{dl}, k}^{\rm T}\hat{\mathbf{h}}_{\mathrm{dl}, k}^{*}\right]\right|^2},\\
&\mathrm{BI}_{k}{\triangleq \rho\kappa_{k}\mathrm{Var}\left(\hat{\mathbf{h}}_{\mathrm{dl}, k}^{\rm T}\hat{\mathbf{h}}_{\mathrm{dl}, k}^{*}\right)},\\
&\mathrm{CI}_{k}{\triangleq \rho\kappa_{k}\Elr{\abs{\tilde{\mathbf{h}}_{\mathrm{dl}, k}^{\rm T}\hat{\mathbf{h}}_{\mathrm{dl}, k}^{*}}^2}},\\
&\mathrm{MI}_{k}{\triangleq\sum\nolimits_{n \neq k}{\rho\kappa_{n}}\Elr{\abs{\mathbf{h}_{\mathrm{dl}, k}^{\rm T}\hat{\mathbf{h}}_{\mathrm{dl}, n}^{*}}^2}};
\end{align}
\end{subequations}
and they are respectively the array gain, the variance of the beamforming uncertainty, self-interference due to calibration imperfections, and multi-\gls{UE} interference, all at the $k$-th \gls{UE}. 
{We note that the \gls{SINR} in~\eqref{eq: SINR_k} can be evaluated for any channel model.}
However, to obtain an intuitive understanding of the effects of imperfect calibration on the downlink \gls{SINR}, a mathematically tractable closed-form expression of~\eqref{eq: SINR_k} is required. This further depends on the underlying distributions of the channels. Due to the multiplicative nature of the calibration coefficients and the composite channels, we consider the following simplifications.
\par
First of all, we assume perfect \gls{CSI} for the repeater to the \gls{BS}~(i.e., $\mathbf{g}$).\footnote{\edit{This assumption, although optimistic, is fairly reasonable in practical scenarios  as typically the \ac{BS} and the repeater would be deployed  at elevated locations (e.g., rooftops) and within \ac{LoS} of one another. 
}}
Then, to capture the effects of uplink channel estimation errors, we consider $\hat{\bar{\mathbf{h}}}_{\mathrm{ul}, k}\sim\mathcal{CN}(\mathbf{0}_{N}, \mathbf{K}_{\bar{\mathbf{h}}_k})$ and $\tilde{\bar{\mathbf{h}}}_{\mathrm{ul}, k}\sim\mathcal{CN}(\mathbf{0}_{N}, \boldsymbol{\Sigma}_{\bar{\mathbf{h}}_k}-\mathbf{K}_{\bar{\mathbf{h}}_k})$, where $\bar{\mathbf{h}}_{\mathrm{ul}, k}\sim\mathcal{CN}(\mathbf{0}_{N}, \boldsymbol{\Sigma}_{\bar{\mathbf{h}}_k})$ and $\boldsymbol{\Sigma}_{\bar{\mathbf{h}}_{ k}} = \beta_{h_k}|t_{\mathrm{U},k}|^2\mathbf{R}_{\rm B}\mathbf{R}_{\rm B}^{\rm H}$.
Similarly, if we assume
$\bar{\mathbf{g}}_{\mathrm{ul},k}\sim\mathcal{CN}(\mathbf{0}_{N}, \boldsymbol{\Sigma}_{\bar{\mathbf{g}}_k})$, we can write  $\hat{\bar{\mathbf{g}}}_{\mathrm{ul},k}\sim\mathcal{CN}(\mathbf{0}_{N}, \mathbf{K}_{\bar{\mathbf{g}}_k})$ and $\tilde{\bar{\mathbf{g}}}_{\mathrm{ul},k}\sim\mathcal{CN}(\mathbf{0}_{N}, \boldsymbol{\Sigma}_{\bar{\mathbf{g}}_k}-\mathbf{K}_{\bar{\mathbf{g}}_k})$, where $\boldsymbol{\Sigma}_{\bar{\mathbf{g}}_{ k}} = \beta_{f_k}|\alpha_u|^2|t_{\mathrm{U},k}|^2\mathbf{R}_{\rm B}\mathbf{g}\mathbf{g}^{\rm H}\mathbf{R}_{\rm B}^{\rm H}$. 
Here, assuming the orthogonal pilot and \ac{MMSE} estimator for the channel estimation, we can define the covariance matrices corresponding to the above variables as follows:
\begin{subequations}
    \begin{align}
        \mathbf{K}_{\bar{\mathbf{h}}_k} &= \rho_{\rm UL}\boldsymbol{\Sigma}_{\bar{\mathbf{h}}_{ k}}\left(\rho_{\rm UL}\boldsymbol{\Sigma}_{\bar{\mathbf{h}}_{ k}} + \sigma_{\rm BS}^2\mathbf{I}_N\right)^{-1}\boldsymbol{\Sigma}_{\bar{\mathbf{h}}_{ k}}\\
        \mathbf{K}_{\bar{\mathbf{g}}_k} &= \rho_{\rm UL}\boldsymbol{\Sigma}_{\bar{\mathbf{g}}_{ k}}\left(\rho_{\rm UL}\boldsymbol{\Sigma}_{\bar{\mathbf{g}}_{ k}} + \sigma_{\rm BS}^2\mathbf{I}_N\right)^{-1}\boldsymbol{\Sigma}_{\bar{\mathbf{g}}_{ k}},
    \end{align}    
\end{subequations}
where $\rho_{\rm UL}$ denotes the uplink transmit power.
Further, we consider that orthogonality between the estimated channels and the error vectors~(which holds for \gls{MMSE} estimators) and the independence between the estimated direct-link channel and the channel via the repeater.
Then $\hat{\mathbf{h}}_{\mathrm{ul}, k}$ is zero-mean \gls{CSCG} random vector with covariance being $\boldsymbol{\Sigma}_{\hat{\mathbf{h}}_{\mathrm{ul}, k}}=\mathbf{K}_{\bar{\mathbf{h}}_k}+\abs{\hat{e}_{\mathrm{R}}}^2\mathbf{K}_{\bar{\mathbf{g}}_k}$.
However, we acknowledge that $\tilde{\mathbf{h}}_{\mathrm{ul}, k}$ is not Gaussian, since it contains products of random variables such as $\tilde{e}_{\mathrm{U}, k}\tilde{\mathbf{E}}_{\rm B}\tilde{\mathbf{h}}_{\mathrm{ul}, k}$. The computation joint distribution of these variables is not mathematically amenable and beyond the scope of this paper. Thus, we approximate $\tilde{\mathbf{h}}_{\mathrm{ul}, k}$ by a \gls{CSCG} random vector
with the covariance being $\boldsymbol{\Sigma}_{\tilde{\mathbf{h}}_{\mathrm{ul}, k}}\triangleq\left(\boldsymbol{\Sigma}_{\bar{\mathbf{h}}_k}\!-\!\mathbf{K}_{\bar{\mathbf{h}}_k}\right)\!+\!\abs{\hat{e}_{\mathrm{R}}}^2\left(\boldsymbol{\Sigma}_{\bar{\mathbf{g}}_k}-\mathbf{K}_{\bar{\mathbf{g}}_k}\right)+\sigma^2_{\mathrm{R}}\boldsymbol{\Sigma}_{\bar{\mathbf{g}}_k}$.
With these, we can now compute the \gls{SINR} in closed form.
\begin{thm}\label{thm: SINR_k}
The downlink \gls{SINR} of the $k$-th \gls{UE} with imperfectly calibrated reciprocity coefficients can be written as
\begin{align}
\label{eq:SINR_error}
\frac{\kappa_{k}\abs{\hat{e}_{\mathrm{U}, k}}^4\abslr{\mathrm{Tr}\left\{\boldsymbol{\Sigma}_{\hat{\mathbf{h}}_{\mathrm{ul}, k}}\abs{\hat{\mathbf{E}}_{\mathrm{B}}}^{2}\right\}}^2}{\sum\limits_{n=1}^{K}\delta_{kn}\mu_{kn}+\sigma_{\mathrm{B}}^2\left(\delta_{kk}\dot{\xi}_{kk}+\sum\limits_{\substack{n\neq k}}\delta_{kn}\xi_{kn}\right)+\dot{\delta}_{kk}\nu_{kk}+\frac{\sigma_{\rm UE}^2}{\rho}},
\end{align}
where the terms are defined as follows:
\begin{subequations}\label{eqn:full}
\begin{align}
&\delta_{kn}=\kappa_{n}\abslr{\hat{e}_{\mathrm{U}, n}}^2\left(\abslr{\hat{e}_{\mathrm{U}, k}}^2+\sigma_{\mathrm{U}, k}^2\right),\label{eq:SINR_subeqs1}\\
&\dot{\delta}_{kk}=\kappa_{k}{\abslr{\hat{e}_{\mathrm{U}, k}}^2\sigma_{\mathrm{U}, k}^2},\label{eq:SINR_subeqs2}\\
&\mu_{kn}=\mathrm{Tr}\left\{\left(\boldsymbol{\Sigma}_{\hat{\mathbf{h}}_{\mathrm{ul}, k}}+\boldsymbol{\Sigma}_{\tilde{\mathbf{h}}_{\mathrm{ul}, k}}\right)\abs{\hat{\mathbf{E}}_{\mathrm{B}}}^{2}\boldsymbol{\Sigma}_{\hat{\mathbf{h}}_{\mathrm{ul}, n}}\abs{\hat{\mathbf{E}}_{\mathrm{B}}}^{2}\right\},\label{eq:SINR_subeqs3}\\
&\xi_{kn}=\mathrm{Tr}\left\{\left(\hat{\mathbf{E}}_{\mathrm{B}}^{*}\boldsymbol{\Sigma}_{\hat{\mathbf{h}}_{\mathrm{ul}, n}}\hat{\mathbf{E}}_{\mathrm{B}}^{\rm T}\right)\odot\left(\boldsymbol{\Sigma}_{\hat{\mathbf{h}}_{\mathrm{ul}, k}}+\boldsymbol{\Sigma}_{\tilde{\mathbf{h}}_{\mathrm{ul}, k}}\right)^{\mathrm{T}}\right\},\label{eq:SINR_subeqs4}\\
&\dot{\xi}_{kk}=\mathrm{Tr}\left\{\hat{\mathbf{E}}_{\mathrm{B}}^{*}\boldsymbol{\Sigma}_{\hat{\mathbf{h}}_{\mathrm{ul}, k}}\odot \left(\hat{\mathbf{E}}_{\mathrm{B}}^{*}\boldsymbol{\Sigma}_{\hat{\mathbf{h}}_{\mathrm{ul}, k}}\right)^{\rm H}\right\}+\xi_{kk},\label{eq:SINR_subeqs5}\\
&\nu_{kk}=\abslr{\mathrm{Tr}\left\{\boldsymbol{\Sigma}_{\hat{\mathbf{h}}_{\mathrm{ul}, k}}\abs{\hat{\mathbf{E}}_{\mathrm{B}}}^{2}\right\}}^2.\label{eq:SINR_subeqs6}
\end{align}
\end{subequations}
\end{thm}
\begin{proof}
See Appendix~\ref{app: SINR_k}.
\end{proof}
\begin{obs}
\label{obs:imperfection}
The part of interference in the downlink \gls{SINR} which arises only due to the effects of calibration imperfections can be shown to be equal to $\{\sum_{n=1}^{K}\kappa_{n}\abslr{\hat{e}_{\mathrm{U}, n}}^2\sigma_{\mathrm{U}, k}^2\mu_{kn}+\sigma_{\mathrm{B}}^2(\delta_{kk}\dot{\xi}_{kk}+\sum\nolimits_{\substack{n\neq k}}\delta_{kn}\xi_{kn})+\sigma_{\mathrm{U}, k}^2\kappa_{k}\abslr{\hat{e}_{\mathrm{U}, k}}^2\nu_{kk}\}$, where the constituent terms explicitly depend on the error variances, $\{\sigma^2_{\mathrm{U}, k},\sigma^2_{\mathrm{B}}\}$, in measuring the calibration coefficients at the \gls{BS} and the \acp{UE}.
\edit{This implies that the calibration errors at the \ac{BS} and \acp{UE} are physically embedded in the transceiver RF chains and hence affect all effective channels.
In contrast, the repeater calibration error variance $\sigma^2_{\mathrm{R}}$ affects only the covariance $\boldsymbol{\Sigma}_{\tilde{\mathbf{h}}_{\mathrm{ul}, k}}$ in $\xi_{kn}$, since it impacts only the channel components that traverse the repeater path.}
\end{obs}
\par
\subsubsection{Special case: No calibration error at the \gls{BS} and \acp{UE}} For this, $e_{\mathrm{U}, k}=\hat{e}_{\mathrm{U}, k}=1, \forall k,$ and $\mathbf{E}_{\mathrm{B}}=\hat{\mathbf{E}}_{\mathrm{B}}=\vI_{N}$. This case subsumes the results presented in~\cite{Yiming_PIMRC}.
\par
While for ease of exposure, herein we consider a single repeater, the preceding analyses can be generalized to multiple repeaters. Conceptually, this extension is straightforward, but it entails substantial additional bookkeeping.
 For example, under independence assumptions across the repeater-related channels, the covariance expressions follow in a straightforward manner.
\section{Numerical Results and Discussions}
In this section, we provide the numerical results of the \gls{SE} analyzed based on the proposed framework.
We consider a circular cell with a radius of 650 [m], where the \gls{BS}, equipped with $N=64$ antennas, is located at the center.
The repeater is located 200 [m] away from the \gls{BS} along the $0^\circ$ direction, and $K=10$ \glspl{UE} are randomly distributed within the cell.
\par
We employ the same channel model as in \cite{2025_Jiannan_repeater} while assuming that a \ac{LoS} \gls{BS}–repeater link always exists by carefully deploying the repeater.
Carrier frequency, bandwidth, and noise density are set to $f_{\rm c} = 6~[\rm GHz]$, $B=20~[\rm MHz]$, and $-174~[\rm dBm/Hz]$, respectively~\cite{2025_Jiannan_repeater}.
The transmit powers of the \gls{BS}, \gls{UE} (for uplink channel estimation), and repeater are set to $\rho=33~[\rm dBm]$, $\rho_{\rm UL} = 23~[\rm dBm]$, and $\rho_{\rm R} = 23~[\rm dBm]$, respectively.
The repeater amplification gain is given by $\alpha = \min(\sqrt{\rho_{\rm R}/(\rho\beta_{\rm g})}, A_{\rm max})$, where $\beta_{\rm g}$ and $A_{\rm max}=90~[\mathrm{dB}]$ denote path-loss of BS-repeater channel and maximum amplification gain, respectively.
In the simulation, following \cite{2018_Emil_reciprocityError}, we model the calibration error as $e_{Z} = 1 + \mathcal{CN}(0, \sigma_{Z}^2), Z\in\{\mathrm{B,U,R}\}$, {and $100$ such random calibration error realizations are considered for \gls{CDF} computations}.
Then, we define the ``ideal calibration'' case as $\sigma_{\rm R}^2=\sigma_{\rm B}^2=\sigma_{\rm U}^2=0$.
For simplicity, we further assume $\alpha_{\rm u}=\min\left(\sqrt{\rho_{\rm R}/(\rho\beta_{\rm g})}, A_{\rm max}\right)$, $t_{\mathrm{U},k}=1,\forall k$, and $\mathbf{R}_{\rm B}=\mathbf{I}_N$.
\begin{figure}[t]
    \centering
    \includegraphics[width=\linewidth]{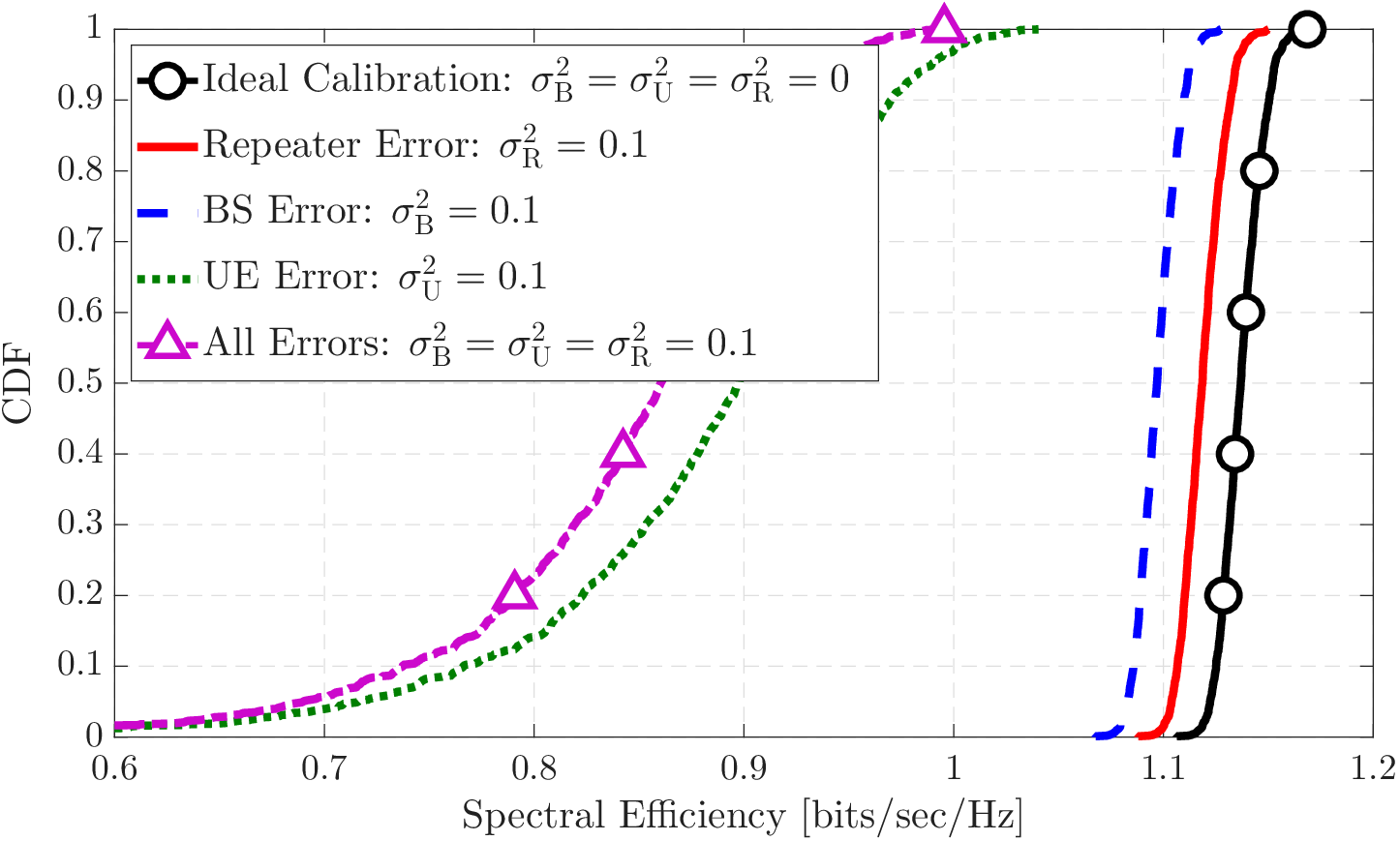}
    \caption{\glspl{CDF} of \gls{SE} with different calibration error variance. In the legend, whenever a variance (e.g., $\sigma_{\rm R}^2=0.1$) is specified, the other variances are set to zero.}
    \label{fig:CDF_SE}
\end{figure}
%
\par
Fig.~\ref{fig:CDF_SE} and Fig.~\ref{fig:SE_variance} illustrate the \gls{CDF} of the \gls{SE} obtained from~\eqref{eq:SINR_error} and the 90th-percentile \ac{SE} performance as a function of the error variance, respectively.
From Fig.~\ref{fig:CDF_SE}, it can be observed that the performance loss caused by the repeater calibration error is smaller than that caused by the other errors.
This is because, as noted in Observation~\ref{obs:imperfection}, the repeater calibration error affects only $\boldsymbol{\Sigma}_{\tilde{\mathbf{h}}_{\mathrm{ul},k}}$ linearly, whereas the \ac{UE} error impacts all interference terms.
When the error variance is small, the performance loss due to the \ac{BS} error is close to that of the repeater error.
However, as shown in both figures, the performance gap between the \ac{BS} error and the repeater error grows with the error variance, even though the losses in both cases increase approximately linearly.
In contrast to $\sigma_{\rm R}^2$, which increases the interference power linearly with respect to $\boldsymbol{\Sigma}_{\tilde{\mathbf{h}}_{\mathrm{ul},k}}$, $\sigma_{\rm B}^2$ increases it linearly with respect to $(\delta_{kk}\dot{\xi}_{kk}+\sum_{n\neq k}\delta_{kn}\xi_{kn})$, which includes $\boldsymbol{\Sigma}_{\tilde{\mathbf{h}}_{\mathrm{ul},k}}$~(see \eqref{eq:SINR_subeqs4}).
As a result, the performance gap between the \ac{BS} error and the repeater error widens as the calibration error variance increases.
\begin{figure}[t]
    \centering
    \includegraphics[width=\linewidth]{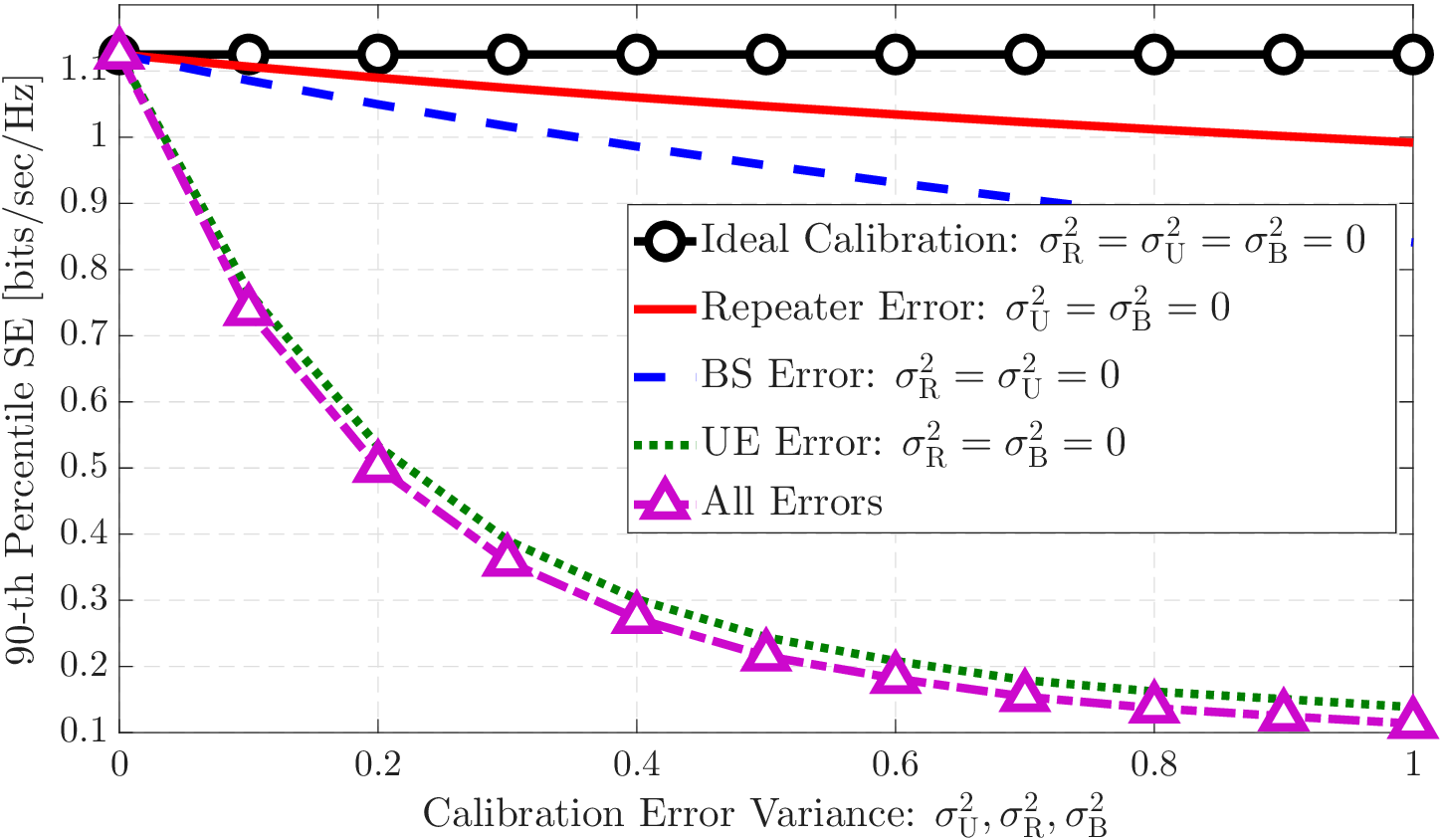}
    \caption{90-th percentile \gls{SE} as a function of calibration error variance.}
    \label{fig:SE_variance}
\end{figure}

\section{Conclusion}
This paper has proposed an \ac{SE} evaluation framework that accounts for imperfect channel estimation and calibration. Compared with previous work that considered only repeater calibration errors \cite{Yiming_PIMRC}, our framework extends the analysis to include BS and UE calibration mismatches.
The numerical results demonstrated that the impact of repeater calibration errors is smaller than that of \ac{BS} and \ac{UE} calibration errors in the considered single-repeater system with Gaussian error modeling. Nonetheless, this indicates that repeater-assisted \ac{MIMO} systems exhibit a certain level of robustness against repeater calibration errors.
Thus, the system can operate with only minimal additional overhead for repeater calibration. However, quantifying the overhead per coherence block, including channel estimation and calibration overhead,  remains to be addressed.

\appendix
\subsection{Derivation of $\mathrm{SINR}_{k}$ for~\cref{thm: SINR_k}.}\label{app: SINR_k}
To compute the array gain, we evaluate the following.
\begin{align}
\mathbb{E}\left[\hat{\mathbf{h}}_{\mathrm{dl}, k}^{\rm T}\hat{\mathbf{h}}_{\mathrm{dl}, k}^{*}\right]
&=\abs{\hat{e}_{\mathrm{U}, k}}^2\mathrm{Tr}\left\{\boldsymbol{\Sigma}_{\hat{\mathbf{h}}_{\mathrm{ul}, k}}\abs{\hat{\mathbf{E}}_{\mathrm{B}}}^{2}\right\},\label{eq: 1}
\end{align}
which follows using~\cref{res: 1}.
Next, we evaluate the variance term involved in $\mathrm{BI}_{k}$.
Using~\cref{res: 1}, we obtain
\begin{align}
&\Elr{\abs{\hat{\mathbf{h}}_{\mathrm{dl}, k}^{\rm T}\hat{\mathbf{h}}_{\mathrm{dl}, k}^{*}}^2}=\abs{\hat{e}_{\mathrm{U}, k}}^4\Elr{\abs{\hat{\mathbf{h}}_{\mathrm{ul}, k}^{\rm T}\abs{\hat{\mathbf{E}}_{\mathrm{B}}}^{2}\hat{\mathbf{h}}_{\mathrm{ul}, k}^{*}}^2}\\
&=\abslr{\hat{e}_{\mathrm{U}, k}}^4\!\left(\abslr{\mathrm{Tr}\left\{\boldsymbol{\Sigma}_{\hat{\mathbf{h}}_{\mathrm{ul}, k}}\abs{\hat{\mathbf{E}}_{\mathrm{B}}}^{2}\right\}}^2\!+\!{\mathrm{Tr}\left\{\left(\boldsymbol{\Sigma}_{\hat{\mathbf{h}}_{\mathrm{ul}, k}}\abs{\hat{\mathbf{E}}_{\mathrm{B}}}^2\right)^{2}\right\}}\right).\notag
\end{align}
Based on this and~\eqref{eq: 1}, we get the final expression for $\mathrm{BI}_{k}$ as
\begin{align}
    \mathrm{BI}_k = \rho\kappa_k|\hat{e}_{\mathrm{U},k}|^4\mathrm{Tr}\left\{ \left( \boldsymbol{\Sigma}_{\hat{\mathbf{h}}_{\mathrm{ul}, k}}|\hat{\mathbf{E}}_{\rm B}|^2 \right)^2 \right\}.
\end{align}
Substituting $\tilde{\mathbf{h}}_{\mathrm{dl}, k}$ from~\eqref{eq: tilde_dl_effective} and due to the mutual uncorrelatedness of the involved terms, the expectation in $\mathrm{CI}_{k}$ becomes
\begin{align}
\Elr{\abs{\tilde{\mathbf{h}}_{\mathrm{dl}, k}^{\rm T}\hat{\mathbf{h}}_{\mathrm{dl}, k}^{*}}^2}&=\Elr{\abslr{\tilde{e}_{\mathrm{U}, k}\hat{e}_{\mathrm{U}, k}^{*}\hat{\mathbf{h}}_{\mathrm{ul}, k}^{\rm T}\hat{\mathbf{E}}_{\mathrm{B}}^{\rm T}\hat{\mathbf{E}}_{\mathrm{B}}^{*}\hat{\mathbf{h}}_{\mathrm{ul}, k}^{\ast}}^2}\notag\\
&+\Elr{\abslr{{e}_{\mathrm{U}, k}\hat{e}_{\mathrm{U}, k}^{*}\hat{\mathbf{h}}_{\mathrm{ul}, k}^{\rm T}\tilde{\mathbf{E}}_{\mathrm{B}}^{\rm T}\hat{\mathbf{E}}_{\mathrm{B}}^{*}\hat{\mathbf{h}}_{\mathrm{ul}, k}^{\ast}}^2}\notag\\
&+\Elr{\abslr{{e}_{\mathrm{U}, k}\hat{e}_{\mathrm{U}, k}^{*}\tilde{\mathbf{h}}_{\mathrm{ul}, k}^{\rm T}\hat{\mathbf{E}}_{\mathrm{B}}^{\rm T}\hat{\mathbf{E}}_{\mathrm{B}}^{*}\hat{\mathbf{h}}_{\mathrm{ul}, k}^{\ast}}^2}\notag\\
&+\Elr{\abslr{{e}_{\mathrm{U}, k}\hat{e}_{\mathrm{U}, k}^{*}\tilde{\mathbf{h}}_{\mathrm{ul}, k}^{\rm T}\tilde{\mathbf{E}}_{\mathrm{B}}^{\rm T}\hat{\mathbf{E}}_{\mathrm{B}}^{*}\hat{\mathbf{h}}_{\mathrm{ul}, k}^{\ast}}^2},\label{eq: 3}
\end{align}
where the cross terms in $\mathbb{E}[\abs{\tilde{\mathbf{h}}_{\mathrm{dl}, k}^{\rm T}\hat{\mathbf{h}}_{\mathrm{dl}, k}^{*}}^2]$ vanish because of the orthogonality of the \gls{MMSE} estimate and the zero-mean property of $\tilde{\mathbf{E}}_{\rm B}$.
Then, using Result~\ref{res: 1}, the first term in~\eqref{eq: 3}, denoted by 
$\mathrm{CI}_{k,1}$, evaluates to 
\begin{align}
\mathrm{CI}_{k,1} \!\!&=\! \abslr{\hat{e}_{\mathrm{U}, k}}^2\sigma_{\mathrm{U}, k}^2(|\mathrm{Tr}\{\boldsymbol{\Sigma}_{\hat{\mathbf{h}}_{\mathrm{ul}, k}}|\hat{\mathbf{E}}_{\mathrm{B}}|^{2}\}|^2\!\!+\!\!{\mathrm{Tr}\{(\boldsymbol{\Sigma}_{\hat{\mathbf{h}}_{\mathrm{ul}, k}}|\hat{\mathbf{E}}_{\mathrm{B}}|^2)^{2}}\}).\notag
\end{align}
Next, the second term, denoted by $\mathrm{CI}_{k,2}$ in~\eqref{eq: 3}, follows 
\begin{align}
\mathrm{CI}_{k,2}
&=p_{\mathrm{U}, k}\sigma^2_{\mathrm{B}}\Bigg[\mathrm{Tr}\left\{\hat{\mathbf{E}}_{\mathrm{B}}^{*}\boldsymbol{\Sigma}_{\hat{\mathbf{h}}_{\mathrm{ul}, k}}\odot \left(\hat{\mathbf{E}}_{\mathrm{B}}^{*}\boldsymbol{\Sigma}_{\hat{\mathbf{h}}_{\mathrm{ul}, k}}\right)^{\rm H}\right\}\notag\\
&\hspace{2ex}+\mathrm{Tr}\left\{\boldsymbol{\Sigma}_{\hat{\mathbf{h}}_{\mathrm{ul}, k}}\odot\left(\hat{\mathbf{E}}_{\mathrm{B}}^{*}\boldsymbol{\Sigma}_{\hat{\mathbf{h}}_{\mathrm{ul}, k}}\hat{\mathbf{E}}_{\mathrm{B}}^{\rm T}\right)^{\mathrm{T}}\right\}\Bigg],
\end{align}
where we define $p_{\mathrm{U}, k}\triangleq\abslr{\hat{e}_{\mathrm{U}, k}}^2(\abslr{\hat{e}_{\mathrm{U}, k}}^2+\sigma_{\mathrm{U}, k}^2)$.
In last step we use~\cref{res: 2}.
Then, we evaluate the third term, denoted by $\mathrm{CI}_{k,3}$ in~\eqref{eq: 3}, as
\begin{align}
\mathrm{CI}_{k,3}
&=p_{\mathrm{U}, k}\mathrm{Tr}\left\{\boldsymbol{\Sigma}_{\tilde{\mathbf{h}}_{\mathrm{ul}, k}}\abs{\hat{\mathbf{E}}_{\mathrm{B}}}^{2}\boldsymbol{\Sigma}_{\hat{\mathbf{h}}_{\mathrm{ul}, k}}\abs{\hat{\mathbf{E}}_{\mathrm{B}}}^{2}\right\}.
\end{align}
The fourth term, denoted by $\mathrm{CI}_{k,4}$, in~\eqref{eq: 3}, is computed as
\begin{align}
\mathrm{CI}_{k,4}
&= p_{\mathrm{U}, k}\sigma^2_{\mathrm{B}}\mathrm{Tr}\left\{\boldsymbol{\Sigma}_{\tilde{\mathbf{h}}_{\mathrm{ul}, k}}\odot\left(\hat{\mathbf{E}}_{\mathrm{B}}^{*}\boldsymbol{\Sigma}_{\hat{\mathbf{h}}_{\mathrm{ul}, k}}\hat{\mathbf{E}}_{\mathrm{B}}^{\rm T}\right)^{\mathrm{T}}\right\}.
\end{align}
Combining $\left\{\mathrm{CI}_{k, i}\right\}_{i=1}^{4}$, we obtain the final expression for the interference due to calibration error at the $k$-th \gls{UE} as $\mathrm{CI}_k = \rho\kappa_k\sum_{i=1}^4\mathrm{CI}_{k,i}.$
This is in part constitutes $\dot{\xi}_{kk}$ in~\eqref{eq:SINR_subeqs5}. 
Finally, we evaluate the multi-\gls{UE} interference. We first write 
\begin{align}
\mathbf{h}_{\mathrm{dl}, k}^{\rm T}\hat{\mathbf{h}}_{\mathrm{dl}, n}^{*}
&=\hat{e}_{\mathrm{U}, k}\hat{e}_{\mathrm{U}, n}^{\ast}\hat{\mathbf{h}}_{\mathrm{ul}, k}^{\rm T}\hat{\mathbf{E}}_{\mathrm{B}}\hat{\mathbf{E}}_{\mathrm{B}}^{*}\hat{\mathbf{h}}_{\mathrm{ul}, n}^{\ast}+\hat{e}_{\mathrm{U}, n}^{*}\tilde{\mathbf{h}}_{\mathrm{dl}, k}^{\rm T}\hat{\mathbf{E}}_{\mathrm{B}}^{*}\hat{\mathbf{h}}_{\mathrm{ul}, n}^{\ast}.\nonumber
\end{align}
Then, we can show that
\begin{align}
\Elr{\abslr{\hat{\mathbf{h}}_{\mathrm{ul}, k}^{\rm T}\hat{\mathbf{E}}_{\mathrm{B}}\hat{\mathbf{E}}_{\mathrm{B}}^{*}\hat{\mathbf{h}}_{\mathrm{ul}, n}}^2}=\mathrm{Tr}\left\{\boldsymbol{\Sigma}_{\hat{\mathbf{h}}_{\mathrm{ul}, k}}\abs{\hat{\mathbf{E}}_{\mathrm{B}}}^{2}\boldsymbol{\Sigma}_{\hat{\mathbf{h}}_{\mathrm{ul}, n}}\abs{\hat{\mathbf{E}}_{\mathrm{B}}}^{2}\right\}.\nonumber
\end{align} 
For the power of the second term, we again substitute $\tilde{\mathbf{h}}_{\mathrm{dl}, k}$ from~\eqref{eq: tilde_dl_effective}. This can be computed similarly to the terms involved in $\mathrm{CI}_{k}$, thus omitted for brevity. 
Combining these terms with some algebra, we obtain the~\ac{SINR} as given in~\cref{thm: SINR_k}.\qed

\subsection{Useful Results}\label{app: results}
\begin{res}\label{res: 1}
For a random vector $\mathbf{a}\sim\mathcal{CN}(\vZ_{N},\mathbf{K})$ and a matrix $\mathbf{Z}$ of compatible dimensions, $\mathbb{E}[\mathbf{a}^{H}\mathbf{Z}\mathbf{a}]=\mathrm{Tr}\left\{\mathbf{K}\mathbf{Z}\right\}$ and $\mathbb{E}[|\mathbf{a}^{H}\mathbf{Z}\mathbf{a}|^2]=|\mathrm{Tr}\left\{\mathbf{K}\mathbf{Z}\right\}|^2+\mathrm{Tr}\left\{\mathbf{K}\mathbf{Z}\mathbf{K}\mathbf{Z}^{\rm H}\right\}$.
\end{res}
\begin{res}\label{res: 2}
For a diagonal matrix $\mathbf{C}$ whose diagonal elements are \gls{iid} random variables with zero mean and variance $\sigma^2_{\mathbf{C}}$; and for any two deterministic matrices $\mathbf{Z}_{1}$ and $\mathbf{Z}_{2}$ of compatible dimensions, it follows that $\mathbb{E}[|\mathrm{Tr}\left\{\mathbf{C}\mathbf{Z}_{1}\right\}|^2]=\sigma^2_{\mathbf{C}}\mathrm{Tr}\left\{\mathbf{Z}_{1}\odot \mathbf{Z}_{1}^{\rm H}\right\}$ and $\Elr{\mathrm{Tr}\left\{\mathbf{C}^{\rm H}\mathbf{Z}_{1}\mathbf{C}\mathbf{Z}_{2}\right\}}=\sigma^2_{\mathbf{C}}\mathrm{Tr}\left\{\mathbf{Z}_{1}\odot \mathbf{Z}_{2}^{\rm T}\right\}$.
\end{res}
These results can be obtained using Isserlis' theorem ~\cite{1918_Isserlis_theorem}.
\ifCLASSOPTIONcaptionsoff
\newpage
\fi
\bibliographystyle{IEEEtran.bst}
\typeout{}
\bibliography{IEEEabrv.bib, Bibliography_list.bib}

\end{document}

%% file: my_preamble.tex
\usepackage{amsmath,amsfonts}
\usepackage{algorithmic}
\usepackage{array}
\usepackage{textcomp}
\usepackage{stfloats}
\usepackage{url}
\usepackage{verbatim}
\usepackage{graphicx}
\def\BibTeX{{\rm B\kern-.05em{\sc i\kern-.025em b}\kern-.08em
		T\kern-.1667em\lower.7ex\hbox{E}\kern-.125emX}} 
\usepackage{balance} 

\usepackage{xargs}                  
 \usepackage{latexsym}
\usepackage[table,xcdraw]{xcolor}
\usepackage{tabularx}
\usepackage{amssymb, amsthm}
\usepackage{placeins}
\usepackage{mathtools}
\usepackage{multirow}
\usepackage{cite, enumerate}
\usepackage[inline,shortlabels]{enumitem}
\usepackage{dsfont}
\usepackage{stackengine,scalerel}
\usepackage{bm}
\usepackage{accents}
\usepackage[normalem]{ulem}
\usepackage{tikz}
\usepackage{circuitikz}
\usepackage{tikz-3dplot}
\usetikzlibrary{shapes.geometric, arrows, positioning, decorations.pathreplacing,decorations.pathmorphing,shapes,backgrounds,patterns,decorations.text}
\usepackage{tkz-euclide}
\tikzset{
    bs/.pic = {                                      
        \draw[line width = 1pt,-round cap] (0,0.4\R) -- (-0.2\R,-0.2\R);
        \draw[line width = 1pt,-round cap] (0,0.4\R) -- (0.2\R,-0.2\R);
        \draw[line width = 1pt] (-0.2\R,-0.2\R) -- (0.133\R,0);
        \draw[line width = 1pt] (0.2\R,-0.2\R) -- (-0.133\R,0);
        \draw[line width = 1pt,-round cap] (-0.133\R,0) -- (0.067\R,0.2\R);
        \draw[line width = 1pt,-round cap] (0.133\R,0) -- (-0.067\R,0.2\R);
        \draw[line width = 1pt,-round cap] (-0.213\R,0.4\R) -- (0.213\R,0.4\R);
        \draw[line width = 1pt,-round cap] \foreach \x in {-0.21, -0.14,...,0.22} {(\x\R,0.4\R) -- (\x\R,0.47\R)};
        \node (dim) at (0,0)  [align=center,minimum width=0.5\R,minimum height=1\R] {};
    },
}
\tikzset{
    user/.pic = {     
        \draw[rounded corners=0.02\R] (-0.09\R,-0.17\R) rectangle (0.09\R,0.17\R) ;                     
        \draw[fill=gray] (-0.08\R,-0.12\R) rectangle (0.08\R,0.12\R);                          
        \draw[rounded corners=0.005\R] (-0.04\R,0.14\R) rectangle (0.04\R,0.15\R);                     
        \draw (0,-0.14\R) circle (0.012\R) ; 
        \node (dim) at (0,0)  [align=center,minimum width=0.2\R,minimum height=0.3\R] {};
    }
}

\tikzset{
  drone/.pic = {
    \begin{scope}[scale=0.3]
      \fill[black] (-1.2,0) to[out=5,in=175] (1.2,0)
                     -- (0.7,-0.3) -- (-0.7,-0.3) -- cycle;
      \foreach \x in {-1.2,1.2} {
        \draw[very thick, black, line cap=round] (\x,0.1) -- (\x,0.4); 
        \draw[very thick, black, line cap=round] (\x-0.5,0.4) -- (\x+0.5,0.4); 
        \fill[black] (\x,0.4) circle (0.06); 
      }
      \fill[black] (-0.5,-0.3) -- (-0.9,-1) -- (-0.7,-1) -- (-0.3,-0.3) -- cycle;
      \fill[black] (0.5,-0.3) -- (0.9,-1) -- (0.7,-1) -- (0.3,-0.3) -- cycle;
      \fill[black] (-0.25,-0.75) rectangle (0.25,-1.25);
      \fill[red] (0,-1.0) circle (0.08);
    \end{scope}
  }
}

\tikzset{
    target/.pic = {
    \begin{scope}[scale=0.6]
        \tikzset{every path/.style={line width=1.3pt}}
        \draw
            (-1,0) -- (-0.75,0)
            .. controls (-0.65,0) and (-0.65,-0.1) .. (-0.6,-0.1)
            arc[start angle=270, end angle=450, radius=0.2cm]
            .. controls (-0.55,-0.1) and (-0.55,0) .. (-0.45,0)
            -- (0.45,0)
            .. controls (0.55,0) and (0.55,-0.1) .. (0.6,-0.1)
            arc[start angle=270, end angle=450, radius=0.2cm]
            .. controls (0.65,-0.1) and (0.65,0) .. (0.75,0)
            -- (1,0) -- (1,0.5) -- (-1,0.5) -- cycle;
        \draw[line width=1pt] (-0.6,0.5) -- (-0.4,0.9) -- (0.4,0.9) -- (0.6,0.5);
        \draw[line width=1pt] (-0.3,0.55) rectangle (0.0,0.85);
        \draw[line width=1pt] (0.1,0.55) rectangle (0.4,0.85);
        \fill (-0.6,-0.1) circle (0.02);
        \fill (0.6,-0.1) circle (0.02);
        \end{scope}
    }
}

\tikzset{
    repeater/.pic={
        \begin{scope}[scale=.78]
            \tikzset{every path/.style={line width=1.3pt}}
            \draw (-.6,-.2) rectangle (.6,.4);
            \draw (-.6,0) -- (-.6,.8);   
            \draw (.6,.4) -- (.6,.8);     

            \draw (-.9,1) -- (-.6,.8) -- (-0.3,1) -- cycle;
	     \draw (.3,1) -- (.6,.8) -- (0.9,1) -- cycle;
         \draw[->,thick, blue] (-.6,1) .. controls (-.4,1) and (-.4,1) .. (.25,1) node [xshift= -2mm, yshift=2mm]{$\alpha_{\sfu}$};
	     \draw[<-,thick, red] (-.4,.8) .. controls (-.4,.8) and (-.4,.8) .. (.4,.8) node [xshift= -4mm, yshift=-2mm]{$\alpha_{\sfd}$};
        \end{scope}
    }
}

\DeclareMathAlphabet\mathbfcal{OMS}{cmsy}{b}{n}

\usepackage{nicematrix}
\usepackage{arydshln}

\usepackage[capitalise,nameinlink]{cleveref}  
\usepackage{prettyref}
\crefformat{equation}{(#2#1#3)} 

\crefname{figure}{Fig.}{Figs.}
\crefname{section}{Sec.}{Secs.}
\crefname{algocf}{Algo.}{Algos.}
\Crefname{algocf}{Algorithm}{Algorithms}
\usepackage{float}
\usepackage[linesnumbered, ruled, vlined]{algorithm2e}
\SetCommentSty{mycommfont}

\SetCommentSty{xCommentSty}
\LinesNumbered
\SetSideCommentRight
\DontPrintSemicolon
\RestyleAlgo{algoruled}

\SetCommentSty{mycommfont}
\SetKwInput{KwInput}{Input}                
\SetKwInput{KwOutput}{Output}            
\SetKwInput{KwInt}{Initialization}            

\newenvironment{myfigure*}
{\begin{figure*}[!t]\centering}
	{\hrule\end{figure*}}

\newcommand{\mt}[1]{\mathsf{#1}}

\newcommand{\Elr}[1]{{ \mathbb{E} \left[#1  \right] }}

\newcommand{\bsr}[1]{{\left[ #1 \right]}}
\newcommand{\bbr}[1]{{\left\{ #1 \right\}}}

\newcommand{\vect}[1]{{\mathbf{ #1}}}

\newcommand{\ocirc}[1]{\ThisStyle{\ensurestackMath{%
  \stackon[.3pt]{\SavedStyle#1}{\SavedStyle\kern.05\LMpt\circ}}}}

\newcommand{\abslr}[1]{{\left\lvert {#1}\right\rvert}}
\newcommand{\abs}[1]{{\big\lvert {#1}\big\rvert}}

\def\vI{\vect{I}}

\def\sfd{{\mt{d}}}

\def\sfu{{\mt{u}}}

\def\vZ{\mathbf{0}}

\newtheorem{thm}{Theorem}

\newtheorem{prop}[thm]{Proposition}

\newtheorem{obs}{Observation}
\newtheorem{res}{Result}